\documentclass[a4paper,11pt]{article}
\usepackage{amsthm}
\usepackage{begin}
\usepackage{algorithm2e}
\usepackage{color}
\usepackage{graphicx}
\usepackage{makeidx}
\usepackage{amsmath}
\usepackage{amssymb}
\usepackage{float}

\titlespacing{\section}{0pt}{3.5ex}{1.5ex}
\titlespacing{\subsection}{0pt}{2.5ex}{0.75ex}
\titlespacing{\subsubsection}{0pt}{2.5ex}{0.75ex}
\titlespacing{\theorem}{0pt}{4.5ex}{0.75ex}

\begin{document}

\newtheorem{theorem}{Theorem}
\newtheorem{lemma}{Lemma}
\newtheorem{definition}{Definition}

\title{Parameterizable Byzantine Broadcast in Loosely Connected Networks}
\author{Alexandre Maurer and S\'{e}bastien Tixeuil\\
UPMC Sorbonne Universit\'{e}s, Paris, France\\
E-mail: Alexandre.Maurer@lip6.fr, Sebastien.Tixeuil@lip6.fr\\
}

\maketitle

\begin{abstract}
We consider the problem of reliably broadcasting information in a multihop asynchronous network, despite the presence of Byzantine failures: some nodes are malicious and behave arbitrarly. We focus on non-cryptographic solutions.

Most existing approaches give conditions for perfect reliable broadcast (all correct nodes deliver the good information), but require a highly connected network. A probabilistic approach was recently proposed for loosely connected networks: the Byzantine failures are randomly distributed, and the correct nodes deliver the good information with high probability. A first solution require the nodes to initially know their position on the network, which may be difficult or impossible in self-organizing or dynamic networks. A second solution relaxed this hypothesis but has much weaker Byzantine tolerance guarantees.

In this paper, we propose a parameterizable broadcast protocol that does not require nodes to have any knowledge about the network. We give a deterministic technique to compute a set of nodes that always deliver authentic information, for a given set of Byzantine failures. Then, we use this technique to experimentally evaluate our protocol, and show that it significantely outperforms previous solutions with the same hypotheses.
\end{abstract}

\vspace{5mm}

\paragraph{Important disclaimer}  These results have NOT yet been published in an international conference or journal. This is just a technical report presenting intermediary and incomplete results. A generalized version of these results may be under submission.

\vspace{7mm}

\section{Introduction}

As modern networks grow larger, they become more likely to fail. Indeed, some nodes can be subject to crashes, attacks, transient bit flips, etc.
Many models for failures and attacks have been studied so far, but the most general one is the \emph{Byzantine} model \cite{LSP82j}: we assume that the failing nodes have a totally arbitrary behavior. In other words, we must anticipate the most malicious strategies they could adopt. This encompasses \emph{all} other possible types of failures, and has important security applications.

In this paper, we study the problem of reliably broadcasting information in a network despite the presence of Byzantine failures. This is a difficult problem, as a single Byzantine node, if not neutralized, can lie to the entire network. Our objective is to design broadcast protocols that prevent or limit the diffusion of malicious informations.

\paragraph{Related works.}

Many Byzantine-robust protocols are based on \emph{cryptography} \cite{CL99c,DFS05c}: the nodes use digital signatures or certificates. Therefore, the correct nodes can assess the validity of received informations and authenticate the sender across multiple hops. However, this approach may not be as general as we want, as the malicious nodes are supposed to ignore some cryptographic secrets: therefore, their behavior is not \emph{completely} arbitrary.
Besides, manipulating cryptographic operations requires important ressources, which may not always be available -- for instance, in sensor networks.
Finally, cryptography requires a trusted infrastructure that initially distributes public and private keys: if this initial infrastructure fails, the whole network fails. Yet, we want to design a solution where \emph{any} component can fail. 
For these reasons, we focus on cryptography-free solutions.

Cryptography-free solutions have first been studied in completely connected networks~\cite{LSP82j,AW98b,MMR03j,MRRS01c,MS03j}: a node can directly communicate with any other node, which implies the presence of a channel between each pair of nodes. Therefore, these approaches are hardly scalable, as the number of channels per node can be physically limited. We thus study solutions in multihop networks, where a node must rely on other nodes to broadcast informations.

A notable class of algorithms tolerates Byzantine failures with either space~\cite{MT07j,NA02c,SOM05c} or time~\cite{MT06cb,DMT11cb,DMT11j,DMT10cd,DMT10ca} locality. Yet, the emphasis of space local algorithms is on containing the fault as close to its source as possible. This is only applicable to the problems where the information from remote nodes is unimportant (such as vertex coloring, link coloring or dining philosophers). Also, time local algorithms presented so far can hold at most one Byzantine node and are not able to mask the effect of Byzantine actions. Thus, the local containment approach is not applicable to reliable broadcast.

It has been shown that, for agreement in the presence of up to $k$ Byzantine nodes, it is necessary and sufficient that the network is $(2k+1)$-connected, and that the number of nodes in the system is at least $3k+1$ \cite{D82j}. Also, this solution assumes that the topology is known to every node, and that nodes are scheduled according to the synchronous execution model.
Both requirements have been relaxed in \cite{NT09j}: the topology is unknown and the scheduling is asynchronous. Yet, this solution retains $2k+1$ connectivity for reliable broadcast and $k+1$ connectivity for detection (the nodes are aware of the presence of a Byzantine failure). In sparse networks such as a grid (where a node has at most four neighbors), both approaches can cope only with a single Byzantine node, independently of the size of the grid. 

Another existing approach is based, not on connectivity, but on the fraction of Byzantine neighbors per node. Broadcast protocols have been proposed for nodes organized on a grid \cite{K04c,BV05c}. However, the wireless medium typically induces much more than four neighbors per node, otherwise the broadcast does not work. Both approaches are based on a local voting system, and perform correctly if every node has strictly less than a $1/4$ fraction of Byzantine neighbors. This result was later generalized to other topologies \cite{PP05j}, assuming that each node knows the global topology. Again, in loosely connected networks, this constraint on the proportion of Byzantine nodes in any neighborhood may be difficult to assess.

All aforementioned results rely on strong \emph{connectivity} and Byzantine proportions assumptions in the network. In other words, tolerating more Byzantine failures requires to increase the connectivity, which can be a heavy constraint in a large network.
To overcome this problem, a probabilistic approach for reliable broadcast has been proposed in \cite{CtrZ}. In this setting, the distribution of Byzantine failures is assumed to be random. This hypothesis is realistic if we consider that each node has a given probability to fail or to be corrupted by an adversary. Another possible use of such a setting is structured overlay networks that make use of distributed hash tables, where the identifier (and thus location) of a node joining the network is attributed randomly: Byzantine nodes joining the network are thus randomly located in the overlay. Also, as we initially accepted that some node may become Byzantine, we can accept that a small minority of correct nodes are fooled by the Byzantine nodes.

With these assumptions, the network can tolerate a number of Byzantine failures that largely exceeds its connectivity \cite{CtrZ}.
This approach was recently generalized to tolerate a constant fraction of Byzantine nodes in a potentially infinite network with a bounded connectivity \cite{Scalbyz}.
However, both solutions require a global view of the network: each node must know its position in the network, in order to join particular subsets called \emph{control zones}.
This may be difficult or impossible in many types of networks, such as self-organized wireless sensor networks or peer-to-peer overlays.
This requirement was removed in \cite{Trig}, at the cost of a much lesser number of tolerated Byzantine than \cite{CtrZ,Scalbyz}.

\paragraph{Our contribution.}

In this paper, we elaborate on~\cite{Trig} and provide a cryptography-free, Byzantine-robust protocol adaptated to loosely connected network, that does \emph{not} require the nodes to initially know their position on the network. Our proposal is parameterizable in such a way that Byzantine tolerance of \cite{Trig} is significantely improved (for example, our experiments show that four times more Byzantine nodes can be tolerated when 0.99 communication reliability probability is required).

The paper is organized as follows:
\begin{itemize}
\item In Section~\ref{sec_desc}, we describe our protocol and give the algorithm executed by each correct node.
\item In Section~\ref{sec_prop}, we present a deterministic technique to determine a set of nodes that deliver the good information in any possible execution.
\item In Section~\ref{sec_exp}, we use this technique to experimentally evaluate the performances of our protocol.
\end{itemize}

\section{Description of the Protocol}
\label{sec_desc}

In this section, we provide an informal description of our protocol. Then, after stating the notations and hypotheses, we give the local algorithm that each correct node must execute.

\begin{figure*}
\begin{center}
\includegraphics[width=15cm]{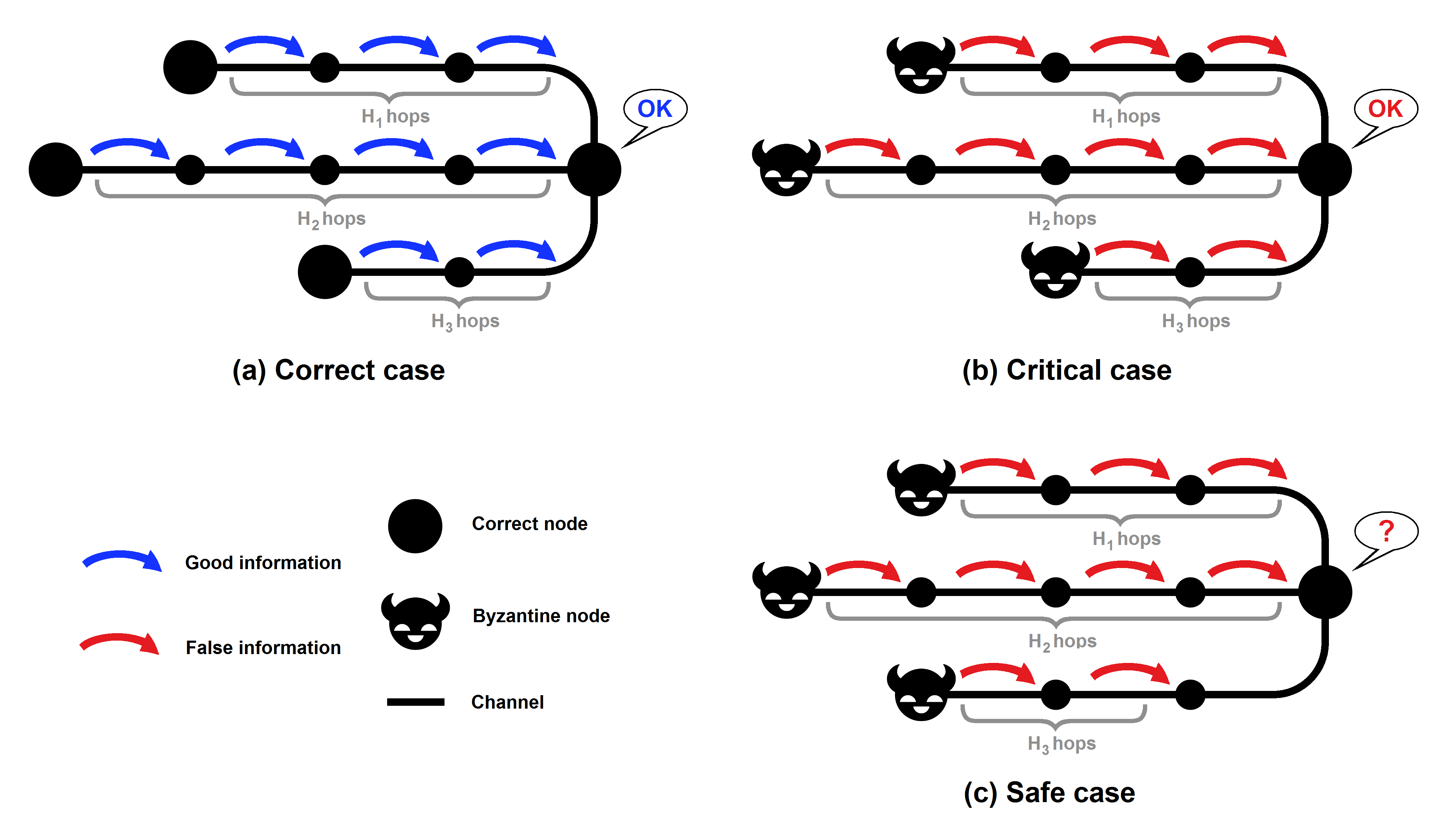}
\caption{Principle of the protocol} 
\label{fig:principle}
\end{center}
\end{figure*}

\subsection{Informal Description}
\label{informal}

The network is described by a set of processes, called \emph{nodes}. Some pairs of nodes are linked by a \emph{channel}, and can send messages to each other: we call them \emph{neighbors}. The network is \emph{asynchronous}: the nodes can send and receive messages at any time.

A particular node, called the \emph{source}, wants to broadcast an information $m$ to the rest of the network. In the ideal case, the source would send $m$ to its neighbors, which will transmit $m$ to their own neighbors -- and so forth, until every node receives $m$. Let us call this an \emph{unsecured broadcast}. In our setting however, some nodes can be malicious (\emph{Byzantine}) and broadcast false informations in the network. Of course, a correct node cannot know whether a neighbor is Byzantine or not.
In the following, we say that a correct node \emph{delivers} an information when it definitively considers that this information comes from the source.

To limit the diffusion of false informations, we introduce the following mechanism:
to deliver an information, a node must receive it through several node-disjoint paths, each path having a maximal length. For instance, in Figure~\ref{fig:principle}-a, the information must be received through \dots

\begin{itemize}
\item a first path of at most $H_1 = 3$ hops
\item a second path of at most $H_2 = 4$ hops
\item a third path of at most $H_3 = 2$ hops
\end{itemize}

This specific setting of the protocol can be described by the triplet $(H_1,H_2,H_3) = (3,4,2)$. To be more general, a setting of the protocol is described by a $n$-uplet $(H_1,\dots,H_n)$.

The underlying idea is as follows: if the Byzantine nodes are sufficiently
spaced, they will never manage to broadcast false informations. Indeed, to deliver a false information in setting $(3,4,2)$, a node must be located at \dots

\begin{itemize}
\item at most $H_1 = 3$ hops from a first Byzantine node
\item at most $H_2 = 4$ hops from a second Byzantine node
\item at most $H_3 = 2$ hops from a third Byzantine node
\end{itemize}

This critical configuration is illustrated in Figure~\ref{fig:principle}-b. However, if (for instance) the third Byzantine node is located at \emph{more} than $2$ hops, the false information will never be delivered. This is illustrated in Figure~\ref{fig:principle}-c. This intuitive idea is demonstrated further in Theorem~\ref{thsafe}.

This protocol is a generalization of the protocol proposed in \cite{Trig}, which corresponds to the settings $(1,H)$, $H$ being a parameter. The condition for safety was the following: the minimal distance between two Byzantine nodes is at least $H+2$ hops. Here, the general condition for safety is more complex to evaluate. However, we show in this paper that the performance of \cite{Trig} can be significantely improved using this more general approach (see Section~\ref{sec_exp}).

\subsection{Detailed description}

Let us precise the notations and hypotheses, and give the local algorithm executed by each correct node.

\paragraph{Notations and Hypotheses}

Let $(G,E)$ be a non-oriented graph representing the topology of the network. $G$ denotes the \emph{nodes} of the network. $E$ denotes the \emph{neighborhood} relationship. A node can only send messages to its neighbors. Some nodes are $correct$ and follow the protocol described thereafter. We consider that all other nodes are totally unpredictable (or \emph{Byzantine}) and may exhibit an arbitrary behavior.
We exclude the case where the source is Byzantine, which obviously cannot be correct.

We consider an asynchronous message passing network: any message sent is eventually received, but it can be at any time. We assume that, in an infinite execution, any process is activated inifinitely often. However, we make no hypothesis on the order of activation of the processes. Finally, we assume local topology knowledge: when a node receives a message from a neighbor $p$, it knows that $p$ is the author of the message. Therefore, a Byzantine node cannot lie about its identity to its direct neighbors. This model is referred to as the ``oral'' model in the literature.

\paragraph{Preliminaries}

The setting of the protocol is described by a $n$-tuple of integers $(H_1,\dots,H_n)$, known by all correct nodes. These values, once settled, should be considered as an inherent part of the protocol: they are ``hard-coded'' with the rest of the algorithm. Therefore, they are not concerned by any broadcast or agreement problem. Let $H = max_{i \in \{1,\dots,n\}} H_i$.

The messages exchanged in the protocol are tuples of the form $(m,S)$:
\begin{itemize}
\item $m$ is the information broadcast by the source, or pretending to be it.
\item $S$ is a set containing the identifiers of nodes already visited by the message. It ensures that a node is never visited twice, and that the paths are actually node-disjoint.
\end{itemize}
Each correct node holds a dynamic set $Rec$, where the messages $(m,S)$ received are recorded. In the following, the set $Rec$ of a node $v$ is referred to as $v.Rec$.

Finally, we say that a node $multicasts$ a message when it sends it to all its neighbors.

\paragraph{Local execution of the protocol}

Initially, the source multicasts $(m,\o)$. Then, each correct node executes the following algorithm:

\begin{enumerate}
\item When a message of type $(m,S)$ is received from a neighbor $q$:
\begin{enumerate}
\item If $q$ is the source, deliver $m$ and multicast $(m,\o)$
\item If $q \notin S$ and $card(S) < H$:
\begin{itemize}
\item Add $(m, S \cup \{q\})$ to the set $Rec$
\item Multicast $(m, S \cup \{q\})$
\end{itemize}
\end{enumerate}

\item When there exists $(m,S_1,\dots,S_n)$ such that:
\begin{enumerate}
\item $\forall i \in \{1,\dots,n\}$, we both have $(m,S_i)\in Rec$ and $card(S_i) \leq H_i$
\item $\forall (i,j) \in \{1,\dots,n\}^2$, $S_i \cap S_j = \o$
\end{enumerate}
Then, deliver $m$ and multicast $(m,\o)$.
\end{enumerate}

\section{Protocol Properties}
\label{sec_prop}

In this section, we adopt the point of view of an omniscient external observer, knowing the positions of Byzantine nodes. These positions are, of course, unknown by the correct nodes: this is just a global view of the system.

In \ref{safenet}, we give a condition on the placement of Byzantine nodes to guarantee the \emph{safety} of the network -- that is, no correct node can deliver a false information. Then, in \ref{relnet}, we give a methodology to determine a reliable node set -- that is, a set of nodes that always deliver the good information, in any possible execution. This methodology is used in the next section to evaluate the performances of the protocol, without actually simulating it.

In addition, \ref{tight}, we show that the bounds of our properties are tight.
At last, in \ref{messcomp}, we evaluate the message complexity of our protocol.

\subsection{Network Safety}
\label{safenet}

Let us give a general condition on the placement of Byzantine nodes, ensuring that no correct node ever delivers a false information.
The following theorem is the demonstration of the intuitive idea exposed in \ref{informal}, proving the correctness of our algorithm.
Notice that is does not necessarily ensure that the correct nodes actually deliver the good information: this aspect is studied further in \ref{partrel}.

Fisrt, let us give some preliminary definitions.

\begin{definition}[Path]
A N-hops \emph{path} is a sequence of distinct nodes $(p_0,\dots,p_N)$ such that, $\forall i \in \{1,\dots,N\}$, $p_i$ and $p_{i-1}$ are neighbors.
We say that this path \emph{connects} $p_0$ and $p_N$. This path is \emph{correct} if all its nodes are correct.
\end{definition}

\begin{definition}[Disjoint paths]
\label{disjoint}
Two path $(p_0,\dots,p_N)$ and $(p'_0,\dots,p'_N)$ are \emph{disjoint} if $\{p_1,\dots,p_{N-1}\} \cap \{p'1,\dots,p'_{N-1}\} = \o$.
\end{definition}

This definition is introduced for the needs of the following theorems. Note that we do not claim that disjoint paths are strictly equivalent to \emph{node-disjoint} paths.

\begin{theorem}[Network Safety]

\label{thsafe}
For a given correct node $u$, let $Critical(u)$ be the following proposition: \\There exists \dots
\begin{itemize}
\item $n$ distinct Byzantine nodes $(b_1,\dots,b_n)$
\item $n$ disjoint paths $(\mathcal{X}_1,\dots,\mathcal{X}_n)$ such that, $\forall i \in \{1,\dots,n\}$, $\mathcal{X}_i$ is a path of at most $H_i$ hops connecting $u$ and $b_i$.
\end{itemize}

If, for every correct node $u$, $Critical(u)$ is false, then no correct node can deliver a false information.
\end{theorem}

\begin{proof}
The proof is by contradiction. Let us suppose the opposite: for every correct node $u$, $Critical(u)$ is false, yet at least one correct node delivers a false information.
Let $v$ be the first correct node to deliver a false information.
In the following, we show that $Critical(v)$ is necessarily true, contradicting the previous statement. This contradiction achieves the proof.

Let $m' \neq m$ be the information delivered by $v$.
As $v$ is correct and delivers $m'$, according to the protocol: there exists $(S_1,\dots,S_n)$ such that,
$\forall i \in \{1,\dots,n\}$, $(m',S_i) \in v.Rec$ and $card(S_i) \leq H_i$.

Let us consider a given index $i \in \{1,\dots,n\}$, and let $q_0 = v$. Let $\mathcal{P}_k$ be the following proposition:
\\There exists a path $(q_1,\dots,q_k)$, with $\{q_1,\dots,q_k\} \subseteq S_i$, such that $q_{k-1}$ received $(m',S_i-\{q_1,\dots,q_k\})$ from $q_k \in S_i - \{q_1,\dots,q_{k-1}\}$.
In our notations, $\{q_1,\dots,q_{k-1}\} = \o$ for $k = 1$.

\begin{itemize}

\item First, let us show that $\mathcal{P}_1$ is true.
According to the protocol, the statement $(m',S_i) \in v.Rec$ implies that $v$ received $(m,S_i-\{q_1\})$ from a node $q_1 \in S_i$.
It is actually possible, as $card(S_i - \{q_1\}) \leq H_i - 1 < H$.
So $\mathcal{P}_1$ is true.

\item  Now, let us suppose that $\mathcal{P}_k$ is true, for $k < card(S_i)$.
So $q_k$ sent $(m',S_i-\{q_1,\dots,q_k\})$ to $q_{k-1}$. Let us suppose that $q_k$ is correct. Then, according to the protocol, it implies that 
$q_k$ received $(m,S_i-\{q_1,\dots,q_{k+1}\})$ from a node $q_{k+1} \in S_i-\{q_1,\dots,q_k\}$.
It is actually possible, as $card(S_i - \{q_1,\dots,q_{k+1}\}) \leq H_i - k - 1 < H$. So either $\mathcal{P}_{k+1}$ is true, either $q_k$ is Byzantine.
\end{itemize}

Therefore, by induction:

\begin{itemize}
\item Either there exists an index $k \in \{1,\dots,card(S_i)-1\}$ and a path $(q_1,\dots,q_k)$ such that $q_k$ is Byzantine. Let $b_i = q_k$, and let $\mathcal{X}_i$ be the path $(v,q_1,\dots,q_k)$.

\item Either $\mathcal{P}_{card(S_i)}$ is true, and $q_{card(S_i)}$ sent $(m',S_i-\{q_1,\dots,q_{card(S_i)}\}) = (m',\o)$ to $q_{card(S_i-1)}$.
The node $q_{card(S_i)}$ cannot be the source, as $m' \neq m$. Also, it cannot be a correct node, as $v$ is the first correct node to deliver $m'$.
So $q_{card(S_i)}$ is necessarily Byzantine. Let $b_i = q_{card(S_i)}$, and let $\mathcal{X}_i$ be the path $(v,q_1,\dots,q_{card(S_i)})$.
\end{itemize}

In both cases, the path $\mathcal{X}_i$ connects $v$ and $b_i$ with at most $card(S_i) \leq H_i$ hops.

According to Definition~\ref{disjoint}, for $(i,j) \in \{1,\dots,n\}^2$, the path $\mathcal{X}_i$ and $ \mathcal{X}_j$ are disjoint if $(S_i-\{b_i\}) \cap (S_j-\{b_j\}) = \o$.
As $v$ is correct and delivers $m'$, according to the protocol: $\forall (i,j) \in \{1,\dots,n\}^2$, $S_i \cap S_j = \o$.
Therefore, the paths $(\mathcal{X}_1,\dots,\mathcal{X}_n)$ are disjoint, and the nodes $(b_1,\dots,b_n)$ are distinct. Thus, $Critical(v)$ is true. Thus, the result.

\end{proof}

Also, as $Critical(u)$ requires $n$ distinct Byzantine nodes, a sufficient condition for safety is that the number of Byzantine nodes does not exceeds $n - 1$.

\subsection{Network Reliability}
\label{relnet}

Here, we suppose that the conditions of Theorem~\ref{thsafe} are satisfied: no false information can be delivered by a correct node.
We now give a methodology to construct a reliable node set.

\label{partrel}

\begin{definition}[Reliable node set]
\label{rns}
For a given source and a given set of Byzantine nodes,
a \emph{reliable node set} is a set of correct nodes that always deliver the good information, in any possible execution.
\end{definition}

The following theorem explains how to construct such a set node by node: for a given reliable node set $\mathcal{R}$ and a given correct node $v$, we can determine if $v$ can be added to $\mathcal{R}$ ($\mathcal{R} := \mathcal{R} \cup \{v\}$) -- and so forth, until we obtain the largest possible reliable node set.
To initialize the construction of $\mathcal{R}$, simply notice that the source and its correct neighbors form a reliable node set, according to the protocol.

\begin{theorem}[Construction of a reliable node set]

\label{threl}
Let us suppose that the conditions of Theorem~\ref{thsafe} are satisfied.
Let $\mathcal{R}$ be a reliable node set, and let $v \notin \mathcal{R}$ be a correct node. If there exists \dots
\begin{itemize}
\item $n$ distinct nodes $(r_1,\dots,r_n)$ of $\mathcal{R}$
\item $n$ disjoint correct paths $(\mathcal{X}_1,\dots,\mathcal{X}_n)$ such that, $\forall i \in \{1,\dots,n\}$, $\mathcal{X}_i$ is a path of at most $H_i$ hops connecting $v$ and $r_i$.
\end{itemize}
Then, $\mathcal{R} \cup \{v\}$ is also a reliable node set.
\end{theorem}

\begin{proof}
According to Theorem~\ref{thsafe}, no correct node can deliver a false information. So, if a correct node delivers an information, it is necessary the good information.

Let us consider a given index $i \in \{1,\dots,n\}$.
Let $\mathcal{X}_i = (q_0,\dots,q_M)$, with $q_0 = r_i$ and $q_M = v$. By definition, we have $M \leq Hi \leq H$.
Let us prove the following property $\mathcal{P}_k$ by induction, $\forall k \in \{1,\dots,M\}$:
\\The node $q_k$ eventually receives $(m,\{q_0,\dots,q_{k-2}\})$ from $q_{k-1}$. In our notations, $\{q_0,\dots,q_{k-2}\} = \o$ for $k=1$.

\begin{itemize}

\item First, let us show that $\mathcal{P}_1$ is true.
As $\mathcal{R}$ is a reliable node set, according to Definition~\ref{rns}, the node $r_i \in \mathcal{R}$ eventually delivers the good information.
According to the protocol, it implies that $r_i$ also multicasts $(m,\o)$.
So $q_1$ eventully receives $(m,\o)$ from $q_0$, and $\mathcal{P}_1$ is true.

\item Now, let us suppose that $\mathcal{P}_k$ is true for $k \leq M$.
As $q_{k-1} \notin \{q_0,\dots,q_{k-2}\}$, and $card(\{q_0,\dots,q_{k-2}\}) < M \leq H$, $q_k$ eventually multicasts $\{q_0,\dots,q_{k-1}\}$.
So $q_{k+1}$ eventually receives $\{q_0,\dots,q_{k-1}\}$ from $q_k$, and $\mathcal{P}_{k+1}$ is true.

\end{itemize}

So $\mathcal{P}_M$ is true and the node $q_M = v$ eventually receives $(m,\{q_0,\dots,q_{M-2}\})$ from $q_{M-1}$.
As $q_{M-1} \notin \{q_0,\dots,q_{M-2}\}$ and $card(\{q_0,\dots,q_{M-2}\}) < M \leq H$, $v$ eventually adds $(m,\{q_0,\dots,q_{M-1}\})$ to the set $Rec$.
Let $S_i = \{q_0,\dots,q_{M-1}\}$.

So, $\forall i \in \{1,\dots,n\}$, we have $(m,S_i) \in v.Rec$ and $card(S_i) < H_i$.
Besides, as the paths $(\mathcal{X}_1,\dots,\mathcal{X}_n)$ are disjoint, according to Definition~\ref{disjoint}: $\forall (i,j) \in \{1,\dots,n\}^2$, we have $(S_i - \{r_i\}) \cap (S_j - \{r_j\}) = \o$.
Thus, as the nodes $(r_1,\dots,r_n)$ are distinct, we have $S_i \cap S_j = \o$.
Therefore, according to the protocol, $v$ eventually delivers $m$, and $\mathcal{R} \cup \{v\}$ is a reliable node set.

\end{proof}

\subsection{Bounds Tightness}
\label{tight}

Let us show that the condition for safety (Theorem~\ref{thsafe}) is tight, and that the methodology to construct a reliable node set (Theorem~\ref{threl}) is optimal in a safe network.

\begin{theorem}[Bounds tightness for Theorem~\ref{thsafe}]
If the condition of Theorem~\ref{thsafe} is not satisfied, it is strictly impossible to guarantee network safety.
\end{theorem}

\begin{proof}
Let us suppose the opposite: the condition of Theorem~\ref{thsafe} is not satisfied, yet the network is safe - that is, no correct node can deliver a false information.
As the condition of Theorem~\ref{thsafe} is not satisfied, there exists at least one correct node $u$ such that $Critical(u)$ is true:
\\There exists \dots
\begin{itemize}
\item $n$ distinct Byzantine nodes $(b_1,\dots,b_n)$
\item $n$ disjoint paths $(\mathcal{X}_1,\dots,\mathcal{X}_n)$ such that, $\forall i \in \{1,\dots,n\}$, $\mathcal{X}_i$ is a path of at most $H_i$ hops connecting $u$ and $b_i$.
\end{itemize}

Thus, it is possible that the Byzantine nodes $(b_1,\dots,b_n)$ unanimously multicast $(m',\o)$, with $m' \neq m$. If so, with a reasoning similar to the proof of Theorem~\ref{threl}, we show that $u$ eventually delivers the false information $m'$. Therefore, the network cannot be safe. This contradiction completes the proof.

\end{proof}

\begin{theorem}[Bounds tightness for Theorem~\ref{threl}]
Let us suppose that the network is safe. Then, the set constructed with Theorem~\ref{threl} is the largest possible reliable node set.
\end{theorem}

\begin{proof}
Let us suppose the opposite: the set $\mathcal{R}$ constructed with Theorem~\ref{threl} is not the largest possible reliable node set. Let $\mathcal{R'}$ be the largest possible reliable node set.

Let there be an execution where all the nodes of $\mathcal{R}$ delivered the good information, but no other correct node delivered the good information so far.
Such an execution is possible, as the construction of $\mathcal{R}$ with Theorem~\ref{threl} does not require that any node $u \notin \mathcal{R}$ delivers the good information.

Let $v$ be the first node of $\mathcal{R'} - \mathcal{R}$ that delivers the good information in the following of the execution.
Then, by a reasoning similar to the proof of Theorem~\ref{thsafe}, we show that there must exist \dots 

\begin{itemize}
\item $n$ distinct nodes $(u_1,\dots,u_n)$
\item $n$ disjoint correct paths $(\mathcal{X}_1,\dots,\mathcal{X}_n)$ such that, $\forall i \in \{1,\dots,n\}$, $\mathcal{X}_i$ is a path of at most $H_i$ hops connecting $v$ and $r_i$.
\end{itemize}

\dots such that the nodes $(u_1,\dots,u_n)$ have previously delivered the good information.
As the only correct nodes that have previously delivered the good information are the nodes of $\mathcal{R}$, we have $\{u_1,\dots,u_n\} \subseteq \mathcal{R}$, and the condition of Theorem~\ref{threl} is verified for $\mathcal{R}$ and $v$. So $\mathcal{R} \cup \{v\}$ is a reliable node set, and $\mathcal{R}$ is not the largest reliable node set that can be constructed with Theorem~\ref{threl}. This contradiction completes the proof.
\end{proof}

\subsection{Message complexity}
\label{messcomp}

At last, let us evaluate the message complexity of our protocol. We only consider the case where all nodes are correct, as the Byzantine nodes can send as many messages as they want.

Let $|G|$ be the number of nodes, and let $d$ be the maximal degree of the network -- that is, the maximal number of neighbors for a single node.
According to our protocol, each node $u$ multicasts $(m,\o)$ at most once.
Then, each neighbor of $u$ multicasts $(m,\{u\})$, which makes at most $1 + d$ messages. This process is repeated $n$ times, which makes at most $1 + d + d^2 + \dots + d^n = o(d^n)$ messages.
So, $o(d^n|G|)$ messages are sent by the protocol.

Therefore, if we consider that $d$ and $n$ are bounded, $o(|G|)$ messages are sent, and the message complexity is the same as an unsecured broadcast protocol (see \ref{informal}).

\section{Experimental Evaluation}
\label{sec_exp}

In this section, we experimentally compare the performances of different settings of our protocol. We detail our motivations, then describe the methodology and comment on the results. Finally, we provide a quantitative comparison with existing protocols, and show the improvement.

\subsection{Motivation}

We would like to evaluate the performances of different settings of our protocol, and also compare them with other existing protocols.

As our protocol is designed for loosely connected networks, we choose grid-shaped networks for experimentation (see Figure~\ref{fig:grid}), where each node has at most $4$ neighbors. Very few Byzantine-robust protocols exist for such sparse networks.
We also compare the performances between grid and \emph{torus} networks (a grid without border), to show the influence of border effects.

In order to quantify these performances, we assume a random uniform distribution of Byzantine failures. Our metric is the probability that a node delivers the good information. We choose this probabilistic model for two reasons:

\begin{itemize}
\item As stated in the introduction, it is difficult or impossible to achieve perfect reliable broadcast is sparse networks. Therefore, we release deterministic guarantees, and aim at probabilistic guarantees -- for instance, a node must deliver the good information with a probability of $0.99$.
\item This random model has realistic applications:
in a network, each component has a weak yet positive probability to misbehave. Besides, some Byzantine agents can join a peer-to-peer overlay at any moment, or be introduced in a collection of wireless sensors before their deployment on the field.
\end{itemize}

Simulating a distributed protocol in the presence of Byzantine failures is a far too difficult problem. Indeed, it would implie to make hypotheses on the order of activation of processes, on the order of reception of messages, and above all, on the behavior of Byzantine nodes: how can we guarantee that they actually adopt the worst possible behavior? To bypass these difficulties, we do not directly simulate the protocol, but use the reliability properties of Section~\ref{sec_prop} to evaluate its performances without adding more hypotheses.

\subsection{Methodology}

We perform our evaluation on torus and grid networks.

\begin{definition}[Torus and grid]
A $N \times N$ \emph{torus} (resp. \emph{grid}) network is a network such that:
\begin{itemize}
\item Each node has a unique identifier $(i,j)$ with
$1 \leq i \leq N$ and $1 \leq j \leq N$.
\item Two nodes $(i_{1},j_{1})$ and $(i_{2},j_{2})$ are neighbors if and only if one of these two conditions is satisfied:
\begin{itemize}
\item $i_{1} = i_{2}$ and $ \lvert j_{1}-j_{2} \rvert  = 1$ or $N$ (resp. $1$).
\item $j_{1} = j_{2}$ and $ \lvert i_{1}-i_{2} \rvert = 1$ or $N$ (resp. $1$).
\end{itemize}
\end{itemize}
\end{definition}

\begin{figure}
\begin{center}
\includegraphics[width=4.5cm]{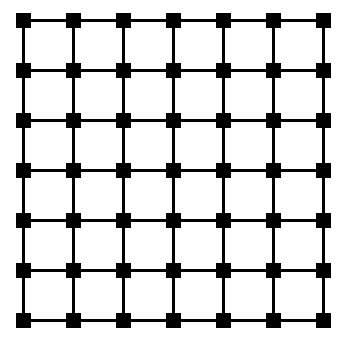}
\caption{A $7 \times 7$ grid network} 
\label{fig:grid}
\end{center}
\end{figure}

These topological identifiers $(i,j)$ are not to mismatch with the identifiers used in the protocol: according to our hypotheses, the nodes do \emph{not} know their position on the grid.

An example of grid network is given in Fig~\ref{fig:grid}.
A torus network can be seen as a grid network with a continuity between the left-right and up-down extremities.

We assume an uniform distribution of Byzantine failures: each node has the same probability $\lambda$ to be Byzantine. We call this probability \emph{Byzantine rate}. We would like to evaluate, for a given Byzantine rate $\lambda$, the probability $P(\lambda)$ for a node to deliver the good information.
For this purpose, we use a Monte-Carlo method:

\begin{itemize}
\item We generate several random distributions of Byzantine nodes, with the Byzantine rate $\lambda$.
\item For each distribution, we randomly choose a source node $s$ and a node $v$. Then, we use Theorem~\ref{thsafe} and Theorem~\ref{threl} to construct a reliable node set (see Definition~\ref{rns}). If $v$ is in the reliable node set, the simulation is a success -- else, it is a failure.
\item With a large number of simulations, the fraction of successes approximates $P(\lambda)$.
\end{itemize}

Il the following, we compare the performances of $4$ different settings of our protocol:
\begin{itemize}
\item Setting $A$: $(1,2)$
\item Setting $B$: $(1,2,5)$
\item Setting $C$: $(1,3,3)$
\item Setting $D$: $(1,2,5,5)$
\end{itemize}

According to Theorem~\ref{thsafe}, we have interest to use the smaller parameters $(H_1,\dots,H_n)$ for our settings. However, if we substract $1$ to \emph{any} parameter of the aforementionned settings, the construction of a reliable node set becomes impossible on grid and torus topologies. Therefore, the choice of these $4$ settings corresponds to a certain optimum.

In particular, setting $A$, of the form $(1,H)$, is the optimal setting for the previous protocol \cite{Trig} on grid and torus topologies. Our experiments show that it is outperformed by other settings.

\subsection{Results}

\begin{figure*}
\begin{center}
\includegraphics[width=18cm]{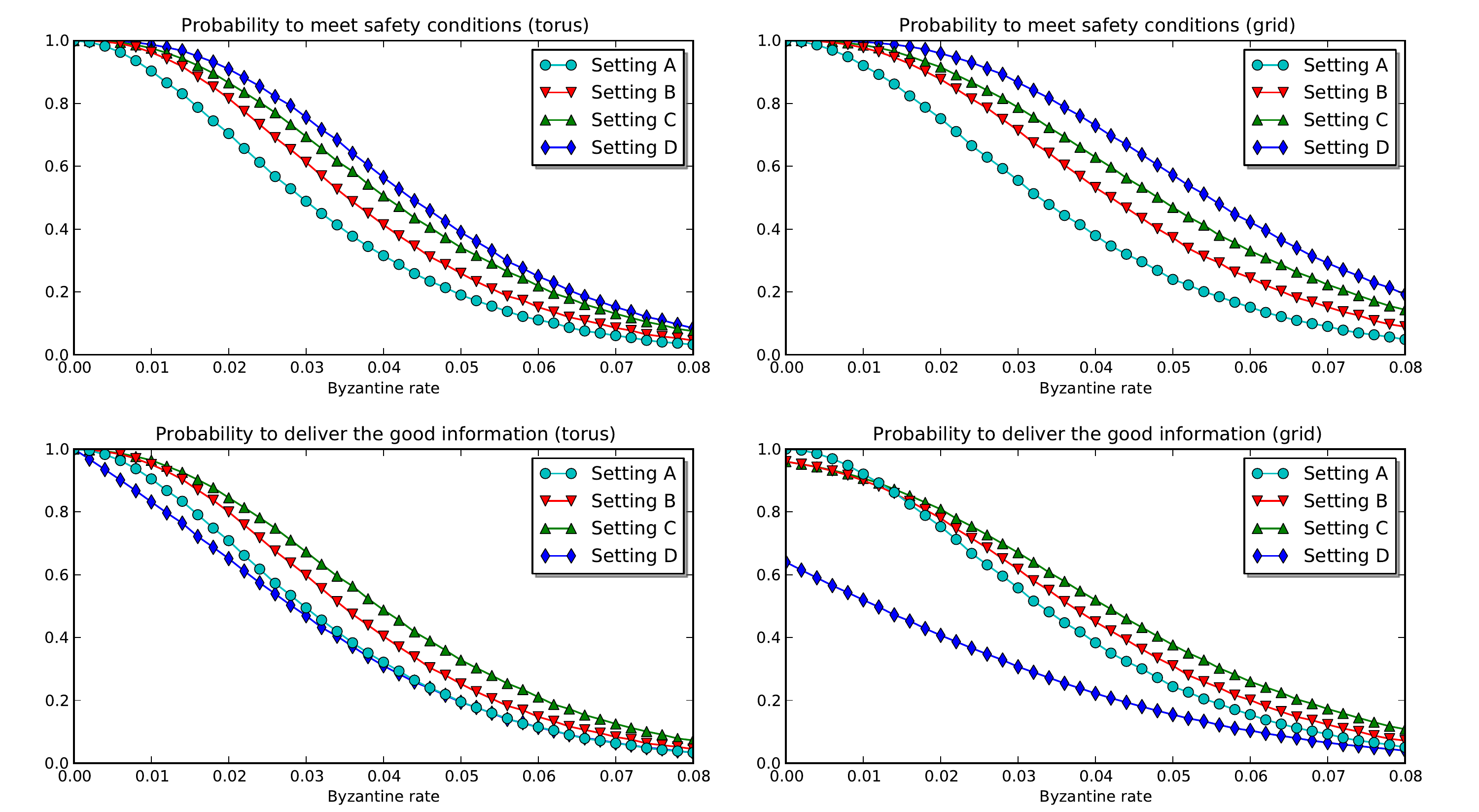}
\caption{Simulation results on a $10 \times 10$ network} 
\label{fig:simu1}
\end{center}
\end{figure*}

\begin{figure*}
\begin{center}
\includegraphics[width=18cm]{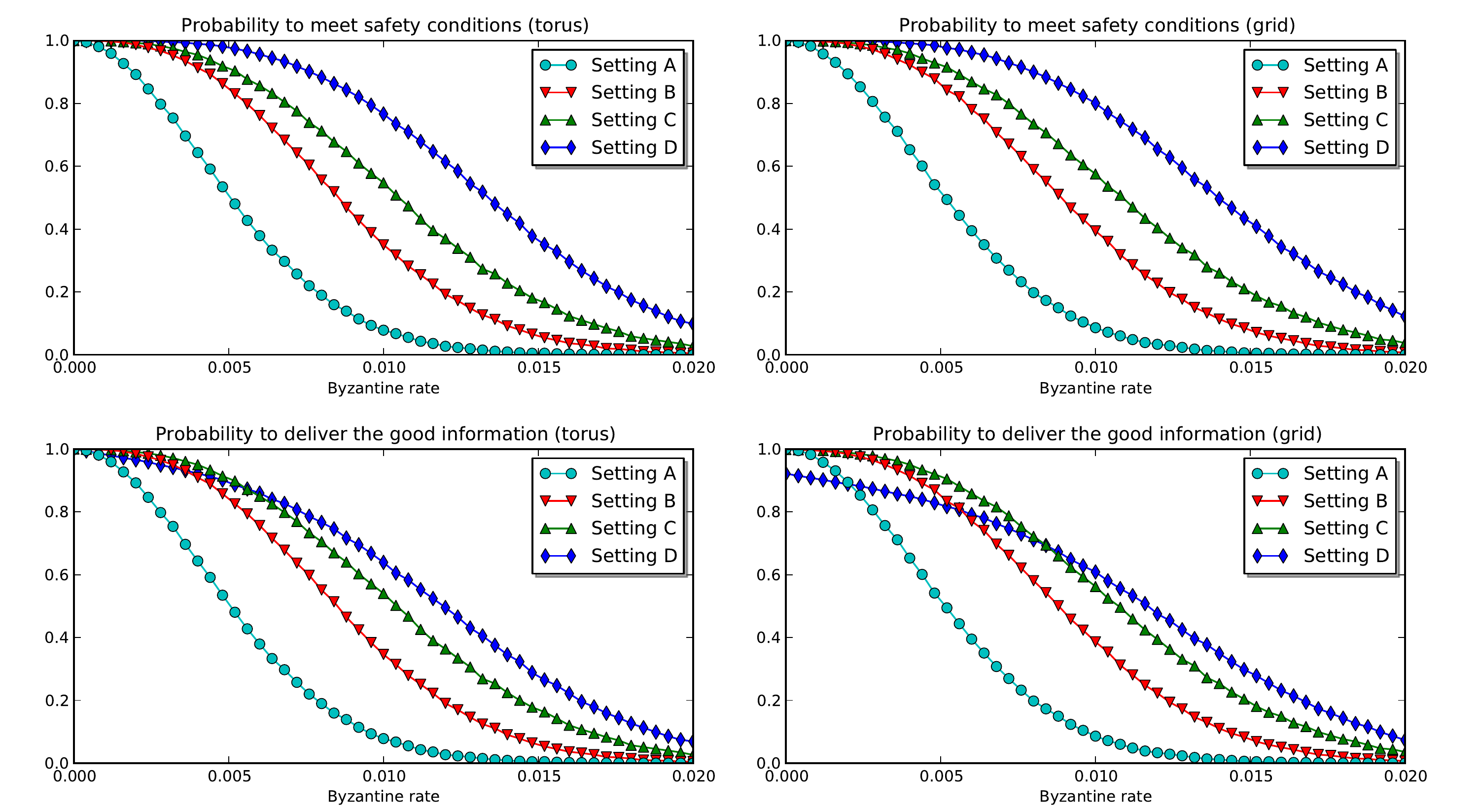}
\caption{Simulation results on a $50 \times 50$ network} 
\label{fig:simu2}
\end{center}
\end{figure*}

\begin{figure*}
\begin{center}
\includegraphics[width=18cm]{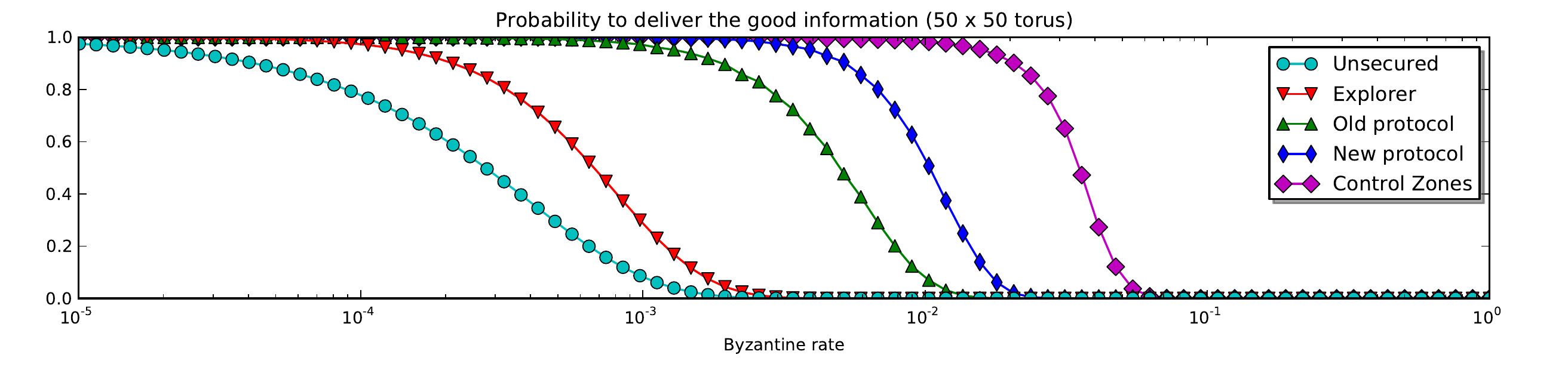}
\caption{Experimental comparison of existing protocols on a logarithmic scale} 
\label{fig:compa}
\end{center}
\end{figure*}

The experimental results are presented in Figure~\ref{fig:simu1} and \ref{fig:simu2}. 
\begin{itemize}
\item We ran simulations on two networks: a ``small'' network of $10 \times 10$ nodes (Figure~\ref{fig:simu1}), and a ``large'' network of $50 \times 50$ nodes (Figure~\ref{fig:compa}).
\item For both networks, we consider the torus topology (left column) and the grid topology (right column).
\item For each case, we represented the probability to meet the safety conditions (no correct node can deliver a false information), then the probability $P(\lambda)$ to deliver the good information.
\end{itemize}

Let us comment on these results. We study the three following aspects: the probability of safety, the border effects and the existence of an optimal setting.

\paragraph{Probability of safety}

The probability to meet safety conditions increases with the complexity of the setting ($2$ paths for setting $A$, $3$ paths for settings $B$ and $C$, $4$ paths for setting $D$). Indeed, as we use more paths, it becomes less likely to have a critical configuration of Byzantine nodes (see Figure~\ref{fig:principle}-b).

This probability decreases with the size of the network: for a given Byzantine rate, a larger network has a greater frequency of critical configurations. Besides, the disparities between the different settings also increase with the size: the results on Figure~\ref{fig:simu1} are more dispersed than on Figure~\ref{fig:simu2}.

At last, the results are slightly better for the grid topology: as the grid is less connected than the torus, a critical configuration is less likely to occur.

\paragraph{Border effects}

When the Byzantine rate equals zero (all nodes are correct), the probability to deliver the good information equals $1$ for the torus topology, but not for the grid topology. Indeed, in the torus, each node has $4$ neighbors, which is sufficient for the $2$, $3$ or $4$ node-disjoint paths required by our settings. However, in the case of the grid, some border nodes have only $3$ or $2$ neighbors (see Figure~\ref{fig:grid}). Therefore:
\begin{itemize}
\item The nodes with $2$ neighbors cannot deliver any information with settings $B$, $C$ and $D$.
\item The nodes with $3$ neighbors cannot deliver any information with setting $D$.
\end{itemize}
Also, we could exclude these problematic border nodes from our statistics, by dividing the probabilities by $P(0)$. If we do so, the results are roughly equivalent for the torus and the grid: the grid offers better chances to meet safety conditions, but the construction of a reliable node set becomes slightly more difficult.

\paragraph{Optimal setting}

There seems to exist an optimal complexity for the setting of the protocol.
Indeed, as we increase the number of paths, we increase the probability to meet safety conditions (according to Theorem~\ref{thsafe}), but we also make the construction of a reliable node set more difficult (according to Theorem~\ref{threl}). Therefore, there is a compromise to find between these two tendancies. This is illustred on the $10 \times 10$ torus, where setting $D$ (the most complex setting) offer the best safety probability, but also the worst communication probability.
Thus, setting $C$ is the best setting for this network.

Besides, the size of the network seems to have an influence on this optimum. Indeed, on the $50 \times 50$ torus, setting $D$ now offers the best communication probability for the Byzantine rates $\lambda > 0.005$.
However, this is not the case for smaller Byzantine rates. Indeed, the neighbors of a Byzantine node cannot deliver the good information with setting $D$, as they have at most $3$ correct neighbors ($4$ correct node-disjoint paths are required). Therefore, setting $C$ achieves the best performances in both networks.

\subsection{Comparison with existing solutions}

At last, let us provide a quantitative comparison with existing solutions. 

According to our hypotheses, there is no initial topology knowledge:
the nodes do not know \emph{a priori} their position on the grid. Therefore, our protocol must be compared with other Byzantine-robust protocols that respect this constraint, yet still work on such sparse networks.
To our knowledge, the only protocols meeting these two conditions
are Explorer \cite{NT09j} and the previous protocol \cite{Trig}.

Let us suppose that we want to guarantee a communication probability $P(\lambda)$ of at least $0.99$ on a $50 \times 50$ torus. Then, we can tolerate a Byzantine rate $\lambda$ of:
\begin{itemize}
\item $4 \times 10^{-6}$ with an unsecured broadcast (see~\ref{informal})
\item $5 \times 10^{-5}$ with Explorer \cite{NT09j}
\item $5 \times 10^{-4}$ with the previous protocol \cite{Trig}
\item $2 \times 10^{-3}$ with our protocol (improvement of factor $4$)
\end{itemize}

These performances are to compare with the Control Zones protocol \cite{CtrZ}, where the nodes are required to know their position on the grid. With this additionnal hypothesis, we can tolerate a Byzantine rate of $8 \times 10^{-3}$. As we can see, releasing topology knowledge still has a cost in term of performances.

All these solutions are represented in Figure~\ref{fig:compa}, on a logarithmic scale.

\section{Conclusion}

In this paper, we proposed a paremetrizable approach for Byzantine Broadcast in loosely connected networks, assuming a minimal number of hypotheses.
We showed that, by manipulating protocol parameters, we can optimize performance in the presence of random Byzantine failures.
However, not using topology knowledge still has a cost in term of Byzantine tolerance.

To go further, it would be interesting to make experimentations on real-life topologies, such as sensor networks, robot networks or peer-to-peer overlays. Also, an motivating open challenge is to obtain theoretical probabilistic guarantees with global network parameters -- diameter, node degree, connectivity -- in order to determine the optimal setting automatically.

\bibliographystyle{plain}
\bibliography{biblio}

\end{document}